\documentclass[conference]{IEEEtran}
\IEEEoverridecommandlockouts

\usepackage[letterpaper, left=0.65in,right=0.65in, top=0.75in,bottom=1in]{geometry}
\usepackage{amsmath,bm, amsfonts}
\usepackage{amssymb}
\usepackage{amsthm}
\usepackage{array}
\usepackage[caption=false,font=normalsize,labelfont=sf,textfont=sf]{subfig}
\usepackage{textcomp}
\usepackage{stfloats}
\usepackage{url}
\usepackage{verbatim}
\usepackage{graphicx}
\usepackage{soul}
\usepackage{algorithm,algpseudocode}
\usepackage{cite}
\usepackage{mathtools}
\usepackage{arydshln}

\def\BibTeX{{\rm B\kern-.05em{\sc i\kern-.025em b}\kern-.08em
    T\kern-.1667em\lower.7ex\hbox{E}\kern-.125emX}}

\usepackage{balance}
\usepackage{xcolor}
\usepackage{array}
\newtheorem{theorem}{Theorem}
\newtheorem{lemma}{Lemma}
\newtheorem{proposition}{Proposition}

\newtheorem{definition}{Definition}

\newtheorem{example}{Example}

\floatname{algorithm}{Algorithm}

\makeatletter
\newcounter{phase}[algorithm]
\newlength{\phaserulewidth}
\newcommand{\setphaserulewidth}{\setlength{\phaserulewidth}}
\newcommand{\phase}[1]{%
  \vspace{-1.25ex}
  \Statex\leavevmode\llap{\rule{\dimexpr\labelwidth+\labelsep}{\phaserulewidth}}\rule{\linewidth}{\phaserulewidth}
  \Statex\strut\refstepcounter{phase}\textit{Phase~\thephase~--~#1}
  \vspace{-1.25ex}\Statex\leavevmode\llap{\rule{\dimexpr\labelwidth+\labelsep}{\phaserulewidth}}\rule{\linewidth}{\phaserulewidth}}
\makeatother
\setphaserulewidth{.7pt}

\usepackage{pgfplots}
\pgfplotsset{compat=newest}
\usetikzlibrary{plotmarks}
\usepackage{grffile}

\newcommand{\ot}[1]{{\textcolor{red}{ot: #1}}}
\newcommand{\dw}[1]{{\textcolor{green}{}}}
\newcommand{\bc}[1]{{\textcolor{black}{#1}}}

\begin{document}

\title{Transversal Toffoli-gate in Hybrid-code System }
\author{
    \IEEEauthorblockN{Dawei~Jiao$^1$, Mahdi Bayanifar$^1$, Alexei Ashikhmin$^2$, Olav Tirkkonen$^1$} 
             \\
              \IEEEauthorblockA{\em$^1$Department of Communications and Networking, Aalto University, Finland\\
              $^{2}$Nokia Bell Labs, Murray Hill, New Jersey, USA
              \\
              Email: \{dawei.jiao, mahdi.bayanifar, olav.tirkkonen\}@aalto.fi} alexei.ashikhmin@nokia-bell-labs.com \\
}
\date{Sep 2025}

\maketitle

\begin{abstract}
    We study the transversality of the Toffoli gate in a hybrid-code system that employs two quantum error correction codes with special structure. 
    We find that a system using a triorthogonal code with its paired code supports a fully transversal implementation of the Toffoli gate.
    Through circuit-level analysis, we prove the transversality of the Toffoli operation in this system.
    Based on this hybrid-code framework, we propose a Toffoli state distillation protocol that does not rely on pre-distilled $\mathbf{T}$-gate magic states.
    In our approach, the Toffoli state is directly distilled layer by layer within the hybrid-code system using only transversal operations.
    Numerical simulations demonstrate that our method uses approximately 50\% fewer qubit resources than previously reported protocols.
\end{abstract}

\section{Introduction}

Fault-tolerant quantum computation relies on carefully designed circuits to prevent the propagation of errors during quantum operations~\cite{Gottesman_1998,eastin2009restrictions}. Transversal gates play an important role in this process, as they act in parallel on each qubit of encoded blocks, thereby naturally preserving the fault-tolerance. Compared with other fault-tolerant approaches, such as magic state and gauge fixing, the transversal gates approach does not require any additional resource costs~\cite{anderson2014codeswitch,bravyi2005magicstate,horsman2012lattice}. Every quantum error correction code (QECC) supports a gate set that can be applied transversally. However, according to the Easting-Knill theorem, it is not possible to have a universal and transversal set of gates simultaneously for a given QECC~\cite{eastin2009restrictions}. 
Some conventional QECCs, such as surface code and Steane code, support
the Clifford gate set transversal realizations, but universal quantum computation additionally requires the inclusion of non-Clifford gates such as the $\mathbf{T}$-gate or the controlled-controlled-phase (CCZ) gate~\cite{yoder2017universal}.

The transversality of some non-Clifford gates, such as $\mathbf{T}$-gate and CCZ gate, have been extensively studied~\cite{bravyi2005magicstate,narayana2020optimality, newman2017limitations}.
In particular, triorthogonal codes have been shown to admit transversal $\mathbf{T}$ and CCZ gates under certain constraints on their logical operators~\cite{bravyi2005magicstate,narayana2020optimality}. 
Besides these gates, there is another well-known non-Clifford gate, called Toffoli-gate.
It is a three-qubit gate that performs a controlled-NOT operation on
a target qubit condition on the state of two control qubits.
The Toffoli gate is frequently used in both reversible classical computation~\cite{toffoli1980reversible} and more complex algorithms, such as Shor's algorithm~\cite{vandersypen2001experimental}. The complexity of implementing Toffoli gate highly depends on the overall complexity of the algorithm.
However, the transversality of the Toffoli gate has rarely been discussed.
Even it is shown that no single QECC can support a transversal Toffoli gate~\cite{newman2017limitations}. 
The use of multiple QECCs in one quantum system in a transversal manner
is investigated in~\cite{bayanifar2025transversality}, where the authors introduced a mirrored code design that enables a transversal CZ gate between two different QECCs.

Motivated by these observations, in this paper we study the transversality of the Toffoli gate using different QECCs. In particular, we identify a scheme that employs two distinct QECCs, enabling a fully transversal implementation of the Toffoli gate. 
To support this approach, we introduce the concept of a \emph{hybrid-code system}, in which multiple QECCs coexist and interact within the same quantum system. In the context of the logical Toffoli-gate, we construct such a system by encoding the control and target qubits using distinct, but structural-related codes.
Specifically, we show that combining a triorthogonal code with its mirrored code 
enables a transversal realization of the Toffoli-gate.

Furthermore, this hybrid-code system can be used as distillation protocol to prepare high-fidelity Toffoli magic state, which is used to realize high-fidelity Toffoli gates. 
Since magic states are key resources in fault-tolerant quantum computation, minimizing the cost of their preparation is crucial.
There exist several  different Toffoli state distillation protocols~\cite{eastin2013distilling,jones2013Toffoli,haah2018Toffoli}, each using different circuit decompositions of Toffoli gate into $\mathbf{T}$ and Clifford gates. In these approaches, the $\mathbf{T}$-gate is implemented using distilled magic states, and the final Toffoli state is obtained by using a large number of magic states. 
In contrast, our distillation protocol is implemented entirely in a transversal manner, eliminating the need for magic states. The Toffoli state is directly distilled layer by layer within the hybrid-code system. Using this approach, we achieve approximately a $50\%$ reduction in qubit resources compared to existing protocols.

The rest of the paper is organized as follows: Sec.~\ref{Sec.Pre} introduces preliminaries on Calderbank-Shor-Steane (CSS) code, transversality and triorthogonal codes. In Sec.~\ref{Sec.ToffoliTr} we show the structure of our hybrid-code system and prove the Toffoli gate transversality. Our new proposed Toffoli state distillation procedure is mentioned in Sec.~\ref{Sec.ToffoliDis}. Finally, in Sec.~\ref{Sec:numerRes} using numerical simulations, we demonstrate that our proposed protocol reduces the qubit number cost compared to other Toffoli state distillation protocols.

\section{Preliminaries}
\label{Sec.Pre}
\subsection{CSS Codes}

CSS quantum error correction codes are a special case of stabilizer codes. 
They can be constructed based on two classical linear codes 
$\mathcal{C}_1\left(n, k_1, d_1\right)$ and $\mathcal{C}_2\left(n, k_2, d_2\right)$, such that $\mathcal{C}_2^{\perp} \subset \mathcal{C}_1$, with 
$\mathcal{C}_2^{\perp}$ the dual space of $\mathcal{C}_2$.
A quantum $[[n, k, d]]$ CSS code $\mathcal{Q} = \rm{CSS}\left(\mathcal{C}_1, \mathcal{C}_2\right)$, is defined as a
$2^k$-dimensional linear subspace of $\mathbb{C}^{2^n}$ with orthonormal basis~\cite{nielsen2002Book}:
\begin{equation}
	\lvert \bm{\psi} \rangle_L = \frac{1}{\sqrt{\left\lvert \mathcal{C}_2^{\perp}\right\rvert} }\sum_{\mathbf{y} \in \mathcal{C}_2^{\perp}}{| \mathbf{x}_\psi+\mathbf{y} \rangle}, \label{Eq:StablStCSS}
\end{equation}
where $\bm{\psi} \in \mathbb{F}_2^k$ and $\mathbf{x}_{\bm{\psi}} = \bm{\psi} \mathbf{A}$, where $\mathbf{A}\in \mathbb{F}_2^{k\times n}$ is a full-rank mapping matrix which is the generator of the quotient group $\mathcal{C}_1 / \mathcal{C}_2^{\perp}$, i.e., $\mathbf{A}\cong\mathcal{C}_1 / \mathcal{C}_2^{\perp}$.
Note that $k=k_1+k_2-n$ and the code minimum distance is  $d = \min\left(d_1, d_2^\perp \right)$.
The generator matrix of $\text{CSS}(C_1,C_2)$ is:
\begin{equation}
    \mathbf{G}^{\mathcal{Q}} = 
    \left[
    \begin{array}{c;{2pt/2pt}c}
        \mathbf{H}\left(\mathcal{C}_2 \right) & \mathbf{0} \\
        \hdashline
        \mathbf{0} & \mathbf{H}\left(\mathcal{C}_1 \right)
    \end{array}
    \right].
\end{equation}

 The dimension of $\mathbf{G}^{\mathcal{Q}}$ is $\left( n-k \right)\times2n$. Thus, with $n$ qubits and $n-k$ generators, $\mathcal{Q} $ encodes $k$ logical qubits. If $\mathcal{C}_1=\mathcal{C}_2=\mathcal{C}$, we have a symmetric CSS code  $\text{CSS}(\mathcal{C},\mathcal{C})$.

\subsection{Transversality}

Each QECC has its transversal gate set. The $\mathbf{U}$ gate transversality for a $[[n,k,d]]$ code is defined as:
\begin{equation}
    \mathbf{U}^{\otimes n}\lvert\bm{\psi}\rangle=\mathbf{U}_L^{\otimes k}\lvert\bm{\psi\rangle}
    \label{Eq.Transversal},
\end{equation}
where $\mathbf{U}_L$ is the logical $\mathbf{U}$ gate, and $\lvert\bm{\psi}\rangle$ is arbitrary logical state.
All CSS codes have CNOT transversality. Symmetric CSS codes have in addition complete Clifford gate transversality. For non-Clifford gate transversality, the most commonly considered one is $\mathbf{T}=\left[\begin{array}{cc}
     1& 0 \\
     0& e^{i\pi/4}
\end{array}\right]$ gate transversality.
Certain CSS codes belonging to a family known as triorthogonal codes have $T$-transversality~\cite{bravyi2005magicstate}. This comes at the cost of losing the transversality of some Clifford gates. 

In addition to transversality within the same code, we can also define  transversality between different codes.
In~\cite{bayanifar2025transversality}, authors found that the restrictions for having transversal logical CNOT from CSS code $\text{CSS} (\mathcal{C}_1,\mathcal{C}_2)$ to  $\text{CSS}(\mathcal{C}_3,\mathcal{C}_4)$  are:
\begin{equation}
    \mathcal{C}_1 / \mathcal{C}_2^{\perp} \cong \mathcal{C}_3 / \mathcal{C}_4^{\perp}, \quad  \mathcal{C}_2^\perp \subseteq \mathcal{C}_4^\perp.
    \label{Eq:CNOTCondition}
\end{equation}
Also, the sufficient conditions for two CSS code $\rm CSS(\mathcal{C}_1,\mathcal{C}_2)$ and  $\rm CSS(\mathcal{C}_3,\mathcal{C}_4)$ to have a transversal logical CZ gate are:
\begin{equation} \label{Eq:CZsuffcond1}
        \mathcal{C}_1 / \mathcal{C}_2^{\perp} \subseteq \mathcal{C}_4, \quad  \mathcal{C}_3 \subseteq \mathcal{C}_2, \quad \mathbf{A B}^T = \mathbf{I},
    \end{equation}
where 
$\mathbf{A}, \mathbf{B} \in \mathbb{F}_2^{k \times n}$,
$\mathbf{A}\cong\mathcal{C}_1 / \mathcal{C}_2^{\perp}$, $\mathbf{B}\cong\mathcal{C}_3 / \mathcal{C}_4^{\perp}$ are two mapping matrices, of the kind discussed after~\eqref{Eq:StablStCSS}. 

If a QECC is both $\mathbf{U}_{a}$ and $\mathbf{U}_{b}$ transversal, it is also $\mathbf{U}_{c}=\mathbf{U}_{a}\mathbf{U}_{b}$ transversal. From 
the definition of transversality~\eqref{Eq.Transversal} we have:
    \begin{align*}
        &\mathbf{U}_a^{\otimes n}\lvert\bm{\psi\rangle}_L=\mathbf{U}_{aL}^{\otimes k}\lvert\bm{\psi\rangle}_L\\
        &\mathbf{U}_b^{\otimes n}\lvert\bm{\psi\rangle}_L=\mathbf{U}_{bL}^{\otimes k}\lvert\bm{\psi\rangle}_L.
    \end{align*}
Combining these, we get 
    \begin{equation}  
    \mathbf{U}_a^{\otimes n}\mathbf{U}_b^{\otimes n}\lvert\bm{\psi\rangle}_L=\mathbf{U}_{aL}^{\otimes k}\mathbf{U}_{bL}^{\otimes k}\lvert\bm{\psi\rangle}_L,
  \label{Prop.Transversality}
\end{equation}
thus we have  $\mathbf{U}_c$ transversality.

\subsection{Triorthogonal code}

Triorthogonal codes are special type of CSS code, which are generated by triorthogonal matrices~\cite{Bravyi_2012}. 
A binary matrix $\mathbf{G}$ of size $m$-by-$n$ is triorthogonal if 
\begin{enumerate}
    \item for all pairs $(a,b)$ that satisfy $1\leq a<b\leq m$, we have 
    \begin{equation}
        \sum_{i=1}^n G_{ai}G_{bi}=0 \mod 2
        \label{Eq:Cond1}
    \end{equation}
    \item for all triples $(a,b,c)$ that satisfy $1\leq a<b<c\leq m$, we have 
    \begin{equation}
        \sum_{i=1}^n G_{ai}G_{bi}G_{ci}=0 \mod 2
    \end{equation}
\end{enumerate}
where $G_{ij}$ indicates the $j$th element in $i$th row of $\mathbf{G}$~\cite{shi2024triorthogonal}.
A triorthogonal matrix $\mathbf{G}$ can be written as:
\begin{equation}
     \mathbf{G}=\left [\begin{array}{cc}
      \mathbf{G}_1 \\
      \hline
      \mathbf{G}_0
 \end{array}\right]
 \label{Eq:TriMatrix}
 \end{equation}
 where all rows in $\mathbf{G}_1$ have odd weight, while rows in $\mathbf{G}_0$ have even weight.
 
A triorthogonal $[[n ,k, d ]]$ CSS code $\mathcal{
Q}^{\rm T} = {\rm CSS}\left( \mathcal{C}_1,  \mathcal{C}_2\right)$ is generated by triorthogonal matrix $\mathbf{G}$ of the form~\eqref{Eq:TriMatrix}, such that  $\mathbf{G}$ and $\mathbf{G}_0^{\perp}$ are the generators of $\mathcal{C}_1$ and $\mathcal{C}_2$, respectively. The generator matrix of the triorthogonal CSS code thus is:
\begin{equation}
    \mathbf{G}^{\mathcal{Q}^T}=\left [\begin{array}{cc}
        \mathbf{G}_0 & \mathbf{0}  \\
         \mathbf{0}& \mathbf{G}^\perp 
    \end{array}\right].
    \label{Eq:TriCodeGen}
\end{equation}

Transversal action of $\mathbf{T}$ gates realize logical $\mathbf{T}$-gates in the triorthogonal codes if certain additional constraints on logical $\mathbf{X}$ operators hold~\cite{narayana2020optimality}.
Triorthogonal codes that support transversal $\mathbf{T}$-gates are referred to as T-transversal triorthogonal codes. In this paper, we use the term \emph{T-triorthogonal codes} for these codes.

\section{Transversal Toffoli-gate}
\label{Sec.ToffoliTr}

In this section, we construct a hybrid-code system that has a transversal logical Toffoli-gate.  
This \emph{hybrid-code system} defined as follows: A system of $k$ parallel logical Toffoli-gates are realized across three code blocks, each protecting $k$ logical qubits. The gate acts in parallel on $k$ triplets of qubits, with one qubit  from each block in each triplet. The first and second block are encoded in T-triorthogonal code $\mathcal{Q}^{\rm T}$. 
We denote mirrored code of $\mathcal{Q}^{\rm T}$ by $\mathcal{Q}^{\rm T_X}$ and use it for the third code block. The concept of mirrored CSS code is defined in ~\cite{bayanifar2025transversality}.
Below, we discuss this concept, and show that this system can achieve transversal Toffoli-gate.

\subsection{Mirrored code $\mathcal{Q}^{\rm T_X}$}

To start with, we focus on the mirrored code $\mathcal{Q}^{\rm T_X}$, emphasizing its features that can be used for providing Toffoli-gate transversality. 

Consider T-triorthogonal code $\mathcal{Q}^{\rm T}=\rm CSS(\mathcal{C}_1,\mathcal{C}_2)$ with generator matrix~\eqref{Eq:TriCodeGen}. The corresponding mirrored code $\mathcal{Q}^{\rm T_X}=\rm CSS(\mathcal{C}_2,\mathcal{C}_1)$ has the generator matrix:
\begin{equation}
    \mathbf{G}^{\mathcal{Q}^{\rm T_X}}=\left [\begin{array}{cc}
        \mathbf{G}^\perp  & \mathbf{0}  \\
         \mathbf{0}& \mathbf{G}_0
    \end{array}\right],
    \label{Eq:MirrorTriCodeGen}
\end{equation}
which can be considered as conjugating all the generators in $\mathcal{Q}^{\rm T}$ with Hadamard operations, exchanging the $\mathbf{X}$ and $\mathbf{Z}$ parts of the generators. 

Defining the transformed gate $\mathbf{T}_{\rm X}=\mathbf{HTH}$, we have:
\begin{lemma}
    If triorthogonal code $\mathcal{Q}^{\rm T}$ has $\mathbf{T}$-gate transversality, 
    its mirrored code $\mathcal{Q}^{\rm T_X}$ has $\mathbf{T}_{\rm X}$ gate transversality. 
    \label{Lemma.TxTr}
\end{lemma}
\begin{proof}
First, by applying $n$ parallel physical Hadamard gates on the generator \eqref{Eq:TriCodeGen} of 
$\mathcal{Q}^{\rm T}$, we get the generator \eqref{Eq:MirrorTriCodeGen} of the mirrored structure code $\mathcal{Q}^{\rm T_X}$, and vice versa. 
Thus, their logical operators and stabilizer groups also can be converted into each other by using $n$ parallel physical Hadamard gates.
Since $\mathcal{Q}^{\rm T}$ has $\mathbf{T}$-gate transversality, which means $n$ parallel physical $\mathbf{T}$-gates achieve $k$ parallel logical $\mathbf{T}$-gates, the mirrored code $\mathcal{Q}^{\rm T_X}$ has $\mathbf{T}_{\rm X}$-gate transversality.
\end{proof}

\subsection{CZ and CNOT Transversality Between $\mathcal{Q}^{\rm T}$ and $\mathcal{Q}^{\rm T_X}$}

As a consequence of the concept of transversal gates between different codes~\cite{bayanifar2025transversality}, we can show that there exists 
CZ and one-direction CNOT gate transversality between the codes 
$\mathcal{Q}^{\rm T}$ and $\mathcal{Q}^{\rm T_X}$.

From~\cite{bayanifar2025transversality},
   it follows that for a pair of codes with mirrored structure, they satisfy the CZ transversal conditions~\eqref{Eq:CZsuffcond1}.
Thus, $\mathcal{Q}^{\rm T}$ and its mirrored code $\mathcal{Q}^{\rm T_X}$ fulfill this condition, they have CZ transversality. 

\begin{lemma}
\dw{Change the order.}
If  $\mathcal{Q}^{\rm T}$ is triorthogonal code and $\mathcal{Q}^{\rm T_X}$ is its mirrored code, the CNOT gate, controlled by $\mathcal{Q}^{\rm T}$ and $\mathcal{Q}^{\rm T_X}$ as target, is transversal. 
    \label{Lemma.CNOTTr}
\end{lemma}
\begin{proof}
    Consider the transversal CNOT gate conditions~\eqref{Eq:CNOTCondition} for the codes $\mathcal{Q}^{\rm T}=\rm CSS(\mathcal{C}_1,\mathcal{C}_2)$ and $\mathcal{Q}^{\rm T_X}=\rm CSS(\mathcal{C}_2,\mathcal{C}_1)$. 
    Since $\mathcal{Q}^{\rm T}$ a triorthogonal code with generator~\eqref{Eq:TriCodeGen}, the classical codes $\mathcal{C}_1$ and $\mathcal{C}_2$ are generated by $\mathbf{G}$ and $\mathbf{G}_0^\perp$, respectively. 
    As $\mathbf{G}$ can be written as~\eqref{Eq:TriMatrix},
\bc{and $\mathbf{G}_0$ is the generator of $\mathcal{C}_2^\perp$.
Thus, 
we have $\mathcal{C}_1/\mathcal{C}_2^\perp\cong\mathbf{G}_1$.}

\bc{A similar analysis applies to the coset $\mathcal{C}_2/\mathcal{C}_1^\perp$, which corresponds to codewords in $\mathbf{G}_0^\perp$ that are not in $\mathbf{G}^\perp$. 
By the triorthogonality condition in~\eqref{Eq:Cond1},
 $\mathbf{G}_0$ is a self-orthogonal matrix, and that each row of $\mathbf{G}_1$ is in $\mathbf{G}_0^\perp$, and $\mathbf{G}_1$ is orthogonal to $\mathbf{G}_0$.
All $2^k$ codewords generated by $\mathbf{G}_1$ are in the coset space $\mathcal{C}_2/\mathcal{C}_1^\perp$. Note that the size of the space generated by  $\mathbf{G}_0^\perp$ and $\mathbf{G}^\perp$ are $2^{n-\rm rank({\mathbf{G}_0})}$ and $2^{n-\rm rank(\mathbf{G}_0)-\rm rank(\mathbf{G}_1)}$, respectively. In addition, the coset $\mathcal{C}_2/\mathcal{C}_1^\perp$ has size $2^{\rm rank(\mathbf{G}_1)}$. This shows that it is exactly the space generated by $\mathbf{G}_1$. }
Thus, we can conclude:
$$\mathcal{C}_2^\perp\subseteq\mathcal{C}_1^\perp \,\,\, \rm and\,\,\,\, \mathcal{C}_1/\mathcal{C}_2^\perp\cong\mathcal{C}_2/\mathcal{C}_1^\perp=\mathbf{G}_1.$$
This mirrored design thus fulfills the CNOT-transversality conditions
\eqref{Eq:CNOTCondition}, with $\mathcal{Q}^{\rm T}$ as control and $\mathcal{Q}^{\rm T_X}$ as target.
\end{proof} 
\bc{Note that, this lemma holds due to the triorthogonality conditions and therefore applies to all triorthogonal codes, regardless of whether they are T-transversal.}

\subsection{Toffoli-gate decomposition}

\begin{figure}
    \centering
    \includegraphics[width=0.6\linewidth]{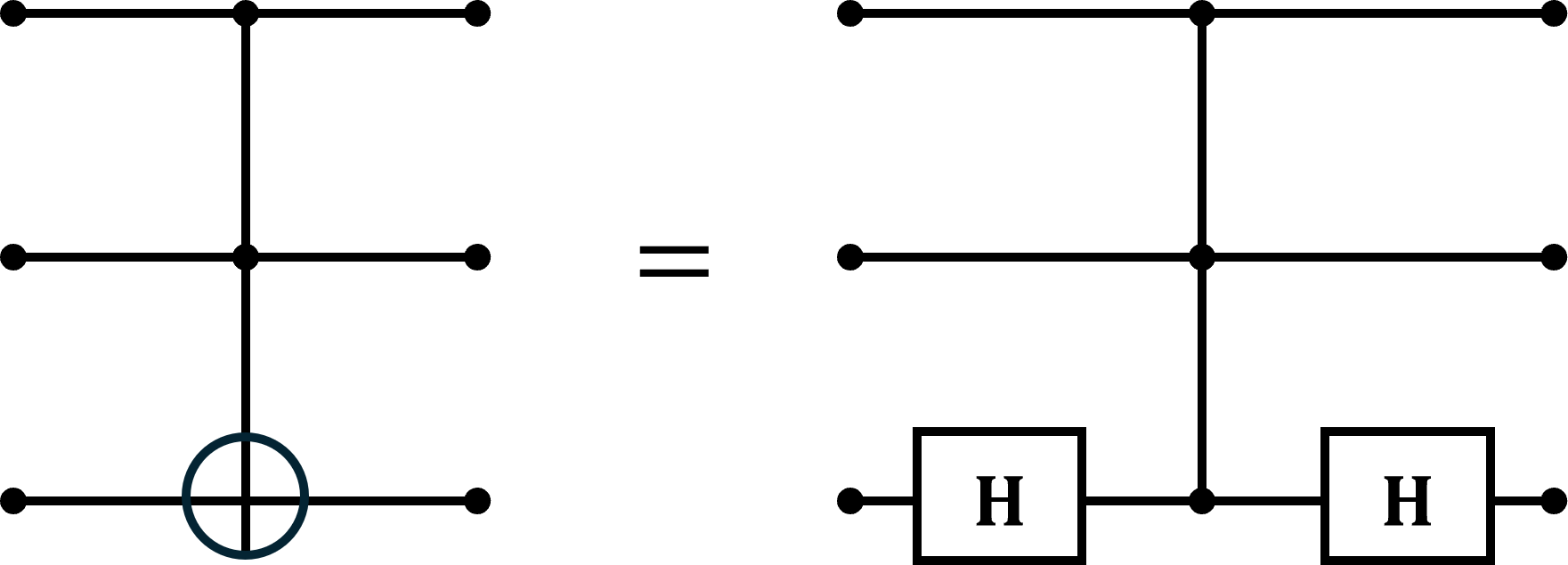}
    \caption{A Toffoli-gate is equivalent to a CCZ gate with target qubit conjugated by Hadamard gates. }
    \label{fig:Toffoli&CCZ}
\end{figure}

A logical Toffoli-gate can be realized by using a logical CCZ gate and conjugating the target qubit with logical Hadamard gates, as depicted in  Fig. \ref{fig:Toffoli&CCZ}. 
The CCZ gate itself can be decomposed into a sequence of $\mathbf{T}$-gates and CNOT gates, leading to a full decomposition of the logical Toffoli-gate into $\mathbf{T}$-gates and logical Clifford gates, shown in Fig. \ref{fig:Toffoli}~\cite{nielsen2002Book}. If T-triorthogonal codes are used, only the CNOT and $\mathbf{T}$-gates can be implemented transversally. The logical Hadamard gates cannot be realized in a transversal manner.
\begin{figure}
    \centering
    \includegraphics[width=0.85\linewidth]{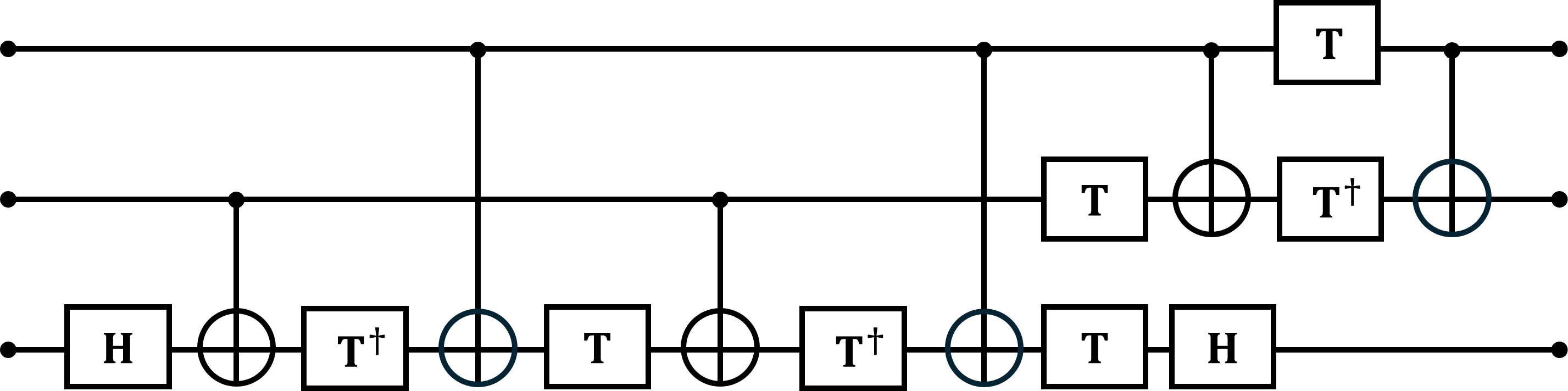}
    \caption{Decomposition of a Toffoli-gate into 6 CNOT gates, 2 Hadamard gates and 7 $\mathbf{T}$ (or $\mathbf{T}^\dagger$) gates.}
    \label{fig:Toffoli}
\end{figure}


\bc{By commuting the first Hadamard gate in Fig.\ref{fig:Toffoli} across all gates acting on the third qubit, each $\mathbf{T}$ gate is transformed into a $\mathbf{T}_{\rm X}$ gate, and CNOT gates are converted into CZ gates. This Hadamard gate then  cancels with the last one, yielding a new Toffoli gate decomposition shown in Fig.\ref{Fig. TxToffoli}. Based on this circuit, we can implement the logical Toffoli gate by encoding the first and second qubits in $\mathcal{Q}^{\rm T}$, and the third qubit in $\mathcal{Q}^{\rm T_X}$. With this encoding, all logical gates in the circuit can be implemented transversally. Therefore, we have:}

\begin{figure}
    \centering
    \includegraphics[width=0.9\linewidth]{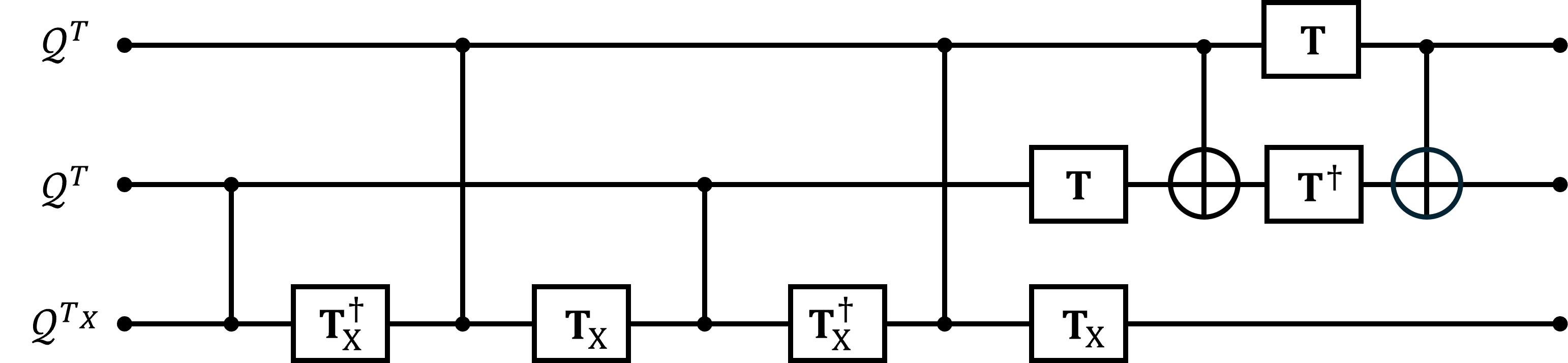}
    \caption{Modified Toffoli-gate decomposition circuit. The first and second qubits encoded in $\mathcal{Q}^{T}$ and the last one encoded in $\mathcal{Q}^{T_X}$.  }
    \label{Fig. TxToffoli}
\end{figure}

\begin{theorem}
    A hybrid-code system with the first and second qubit encoded in $\mathcal{Q}^{\rm T}$ and third qubit encoded in $\mathcal{Q}^{\rm T_X}$
   has Toffoli-gate transversality with the first and second qubit as control qubits and third as target qubit.
\end{theorem}
\begin{proof}
    Since $\mathcal{Q}^{\rm T}$ is a T-triorthogonal CSS code, the first and second qubits, encoded in $\mathcal{Q}^{\rm T}$, have $\mathbf{T}$-gate and CNOT gate transversalities. 
    From Lemma \ref{Lemma.TxTr}, the third qubit, encoded in $\mathcal{Q}^{\rm T_X}$, supports transversal $\mathbf{T}_{\rm X}$. From \cite[Example.2]{bayanifar2025transversality}, it follows that CZ gates are transversal between qubits 1 or 2 and 3.  
    In the Toffoli-gate decomposition circuit shown in Fig. \ref{Fig. TxToffoli}, every operation can thus be applied transversally across three code blocks. Since each component of the circuit support the transversal conditions, from \eqref{Prop.Transversality} it follows that the logical Toffoli-gate can be directly implemented in a transversal manner in this system.
\end{proof}
\dw{Two ways of using it to have universal computation.}

With such a hybrid-code system, two approaches can be used to achieve universal quantum computation. The first approach uses the T-triorthogonal code with its mirrored code. 
In this case, the system implements non-Clifford gate transversality, and universality can be achieved by applying the logical Hadamard gate through an ancilla block~\cite{paetznick2013universal}.    
The second approach exploits the Toffoli gate transversality to perform Toffoli state distillation. Similar to magic state distillation, the distilled Toffoli states serve as ancilla states for implementing logical Toffoli gates in other code systems with transversal Clifford gate sets, such as surface codes or color codes.
We will explain this distillation protocol in the following section.

\section{Toffoli-state distillation}
\label{Sec.ToffoliDis}
As an example application for the proposed hybrid-code system with Toffoli gate transversality,  
we consider a Toffoli state distillation protocol. 
The Toffoli state $\lvert \rm \psi_{Tof}\rangle$ is a three qubit magic state defined by the output of the Toffoli-gate acting on the initial state: 
\dw{define CCX.}
\begin{equation}
    \lvert \rm \psi_{Tof}\rangle=\mathbf{CCX}\lvert +\rangle\lvert +\rangle\lvert 0\rangle,
\end{equation} 
where $\mathbf{CCX}$ indicates Toffoli gate and with the third qubit $\lvert0\rangle$ as the target.
It can be used to implement the action of a Toffoli-gate with a circuit shown as Fig.\ref{fig:ToffoliState}, as an assistant state. 
The lower three qubits in arbitrary state $\lvert a\rangle\lvert b\rangle\lvert c\rangle$ are the input state.
With following three CNOT gates between Toffoli state and input state, the input state are measured in Pauli $\mathbf{X}$ and $\mathbf{Z}$ basis. According to the measurement outcome, the feedback operations, connected by solid lines, are applied. For example, if the measurement on $\lvert c\rangle$ is $-1$, then a Pauli $\mathbf{Z}$ and CZ operators are applied simultaneously. The outcome of left circuit is $\lvert a\rangle\lvert b\rangle\lvert a\cdot b\oplus c\rangle$, which is the same as the outcome of $\lvert a\rangle\lvert b\rangle\lvert c\rangle$ go through Toffoli-gate.

Our Toffoli state distillation protocol is based on the hybrid-code system, which takes low-fidelity Toffoli states as input and outputs distilled, high-fidelity physical Toffoli states.
These high-fidelity states can then be fault-tolerantly encoded into any code used within the computational system. In cases where direct fault-tolerant encoding is not available, a concatenated architecture can be employed, combining the computational code with the hybrid-code system.
There are several ways to fault-tolerantly distill Toffoli-states~\cite{haah2018Toffoli,jones2013Toffoli,eastin2013distilling}. These protocols use distilled magic states to implement logical $\mathbf{T}$-gates, and through specific logical circuits, high-fidelity Toffoli states can be obtained.
Our distillation method exploits Toffoli-gate transversality discussed above, and directly distill the Toffoli state for achieving high-fidelity state.
The distillation protocol can be summarized as follows:

\begin{figure}
    \centering
    \includegraphics[width=0.95\linewidth]{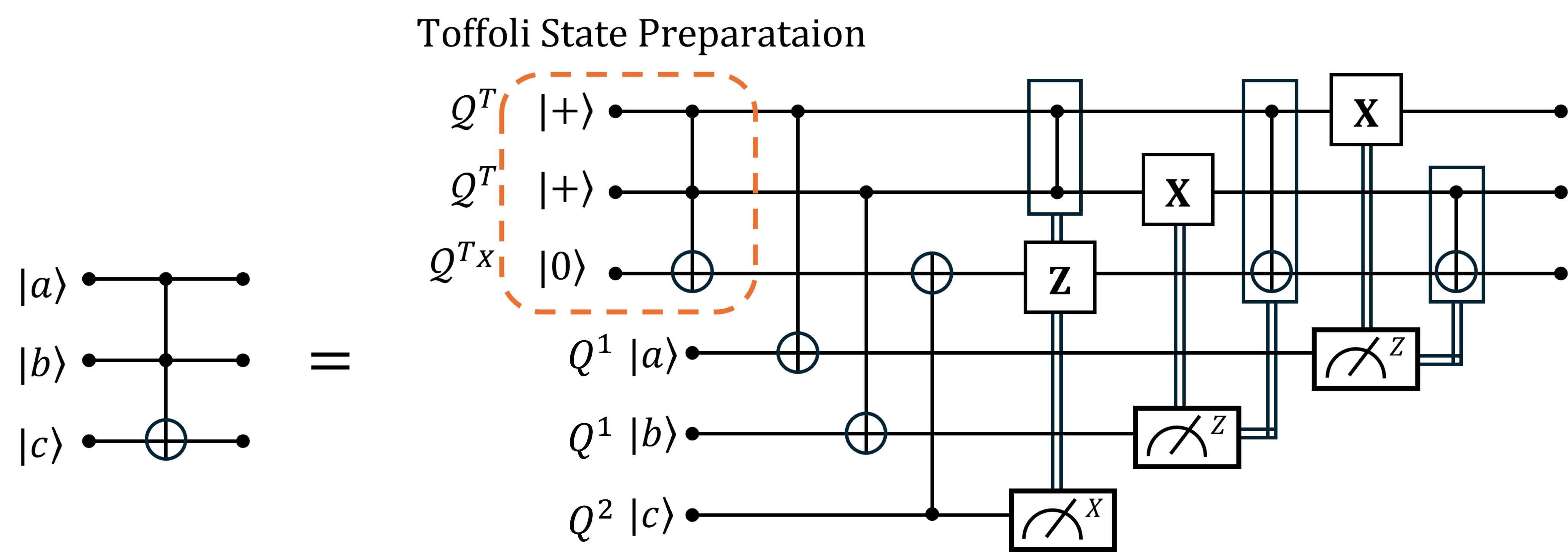}
    \caption{
    Circuit for both Toffoli state distillation and Toffoli gate implementation. When the Toffoli state is used to perform a Toffoli gate, $\mathcal{Q}^1$ and $\mathcal{Q}^2$ correspond to the same code employed in the computational system, such as the surface code. In the Toffoli-state distillation process, the $\mathcal{Q}^1$ and $\mathcal{Q}^2$ are chosen as $\mathcal{Q}^{\rm T}$ and $\mathcal{Q}^{\rm T_x}$ respectively. The orange dashed box outputs an ancillary Toffoli-state. Using CNOTs and measurements, the Toffoli-gate action is teleported to the 3-qubit state  $\lvert a\rangle\lvert b\rangle\lvert c\rangle$. }
    \label{fig:ToffoliState}
\end{figure}

\dw{Using algorithm manner.}
\begin{itemize}
    \item The first layer distillation:
    \begin{itemize}
        \item Three code blocks are encoded in logical $\lvert+\rangle_L\lvert+\rangle_L\lvert0\rangle_L$ with hybrid-code system using $\mathcal{Q}^{\rm T}$, $\mathcal{Q}^{\rm T}$ and $\mathcal{Q}^{\rm T_X}$. 
        \item Transversal Toffoli gates are applied on $n$ physical qubits between three code blocks, and $\mathcal{Q}^{\rm T_X}$ as target. If physical Toffoli gate is not available in quantum hardware, one can use circuit shown in Fig.\ref{Fig. TxToffoli}, with all the element gates can be transversally applied.
        \item Syndrome measurements are applied after transversal Toffoli gate (or after each element gates in Fig.\ref{Fig. TxToffoli}). If non-trivial syndromes are detected, discard the state and start from beginning.
        \item The outcome state is $\lvert \rm \psi_{Tof}\rangle_L^{\otimes k}$, then we decode the logical state and get $k$ distilled physical Toffoli states.
    \end{itemize}
    \item The second and higher layer distillation:
    \begin{itemize}
        \item We prepare three code blocks encoded in logical state  $\lvert+\rangle_L\lvert+\rangle_L\lvert0\rangle_L$ with hybrid-code system.
        \item Prepare $n$ physical distilled Toffoli states from the last layer distillation.
        \item We apply $n$ transversal Toffoli gates between physical qubits in three code blocks. Each transversal Toffoli gate can be realized via circuit shown as Fig.\ref{fig:ToffoliState}.
        \begin{itemize}
            \item First conduct $3n$ transversal CNOT gates between the first layer distilled states and the second layer input states.
        \item The input states are measured in Pauli $\mathbf{X}$ and $\mathbf{Z}$ basis. The feedback transversal operations are applied accordingly. 
        \item Syndrome measurement are applied to detect the errors in these procedures. If outcome is non-trivial syndrome, discard the state and start from the beginning.
        \end{itemize} 
        \item Decoding the result states, we get $k$ next-layer distilled Toffoli states.    
    \end{itemize}
\end{itemize}

Notice, after repeated use of distilled Toffoli-states, we can achieve Toffoli-gates with arbitrarily high fidelity, as long as the physical error rate of initial physical Toffoli-gates is lower than a threshold.

Our Toffoli state distillation protocol can produce a high-fidelity Toffoli state without the need to generate seven magic states that are used to realize the circuit of Fig.~\ref{fig:Toffoli}. This not only reduces the overall resource overhead but also minimizes the number of encoding procedures, as each distilled physical magic state would need to be re-encoded into a different code system.

\section{Numerical Simulation}\label{Sec:numerRes}
As an example, consider using the $[[15,1,3]]$ T-triorthogonal code and its mirrored code for Toffoli-state distillation.  For performance evaluation, we assume noise on the physical Toffoli-gates as the only error source.
\bc{For each code block, the main contribution to non-trivial syndromes are weight-1 and weight-2 errors. For simplicity we assume that after apply physical Toffoli gate, all qubits have the same error probability $p$ to occur a $\mathbf{Z}$ error, the successful rate $p_s$ to conduct distillation can be written as:
\begin{equation*}
    p_s=1-15p-105p^2+O(p^3).
\end{equation*}
When we get trivial syndrome, the outcome state in each code block has error rate $35p^3+O(p^4)$, which the factor 35 comes from the number of non-detectable weight-3 errors in $[[15,1,3]]$ code. This performance is the same as using $[[15,1,3]]$ code for magic state distillation, thus the threshold for $p$ is 14.1\%, the same as 15-to-1 distillation protocol~\cite{bravyi2005magicstate}. Since Toffoli gate is a three qubit gate, and for each qubit it has error threshold $p_\epsilon=14.1\%$, the threshold for Toffoli error rate should be $p_{th}\approx3p_\epsilon(1-p_\epsilon)^2=31.8\%$.
}

\bc{An alternative approach is using another set of T-triorthogonal codes with parameter $[[3k+8,k,2]]$~\cite{Bravyi_2012}. By using this code, one can achieve successful error rate as: $p_s=1-(3k+8)p+O(p^2)$ and output error rate on any one qubit is $(1+3k)p^2$.}

\begin{figure}
    \centering
    \includegraphics[width=0.9\linewidth]{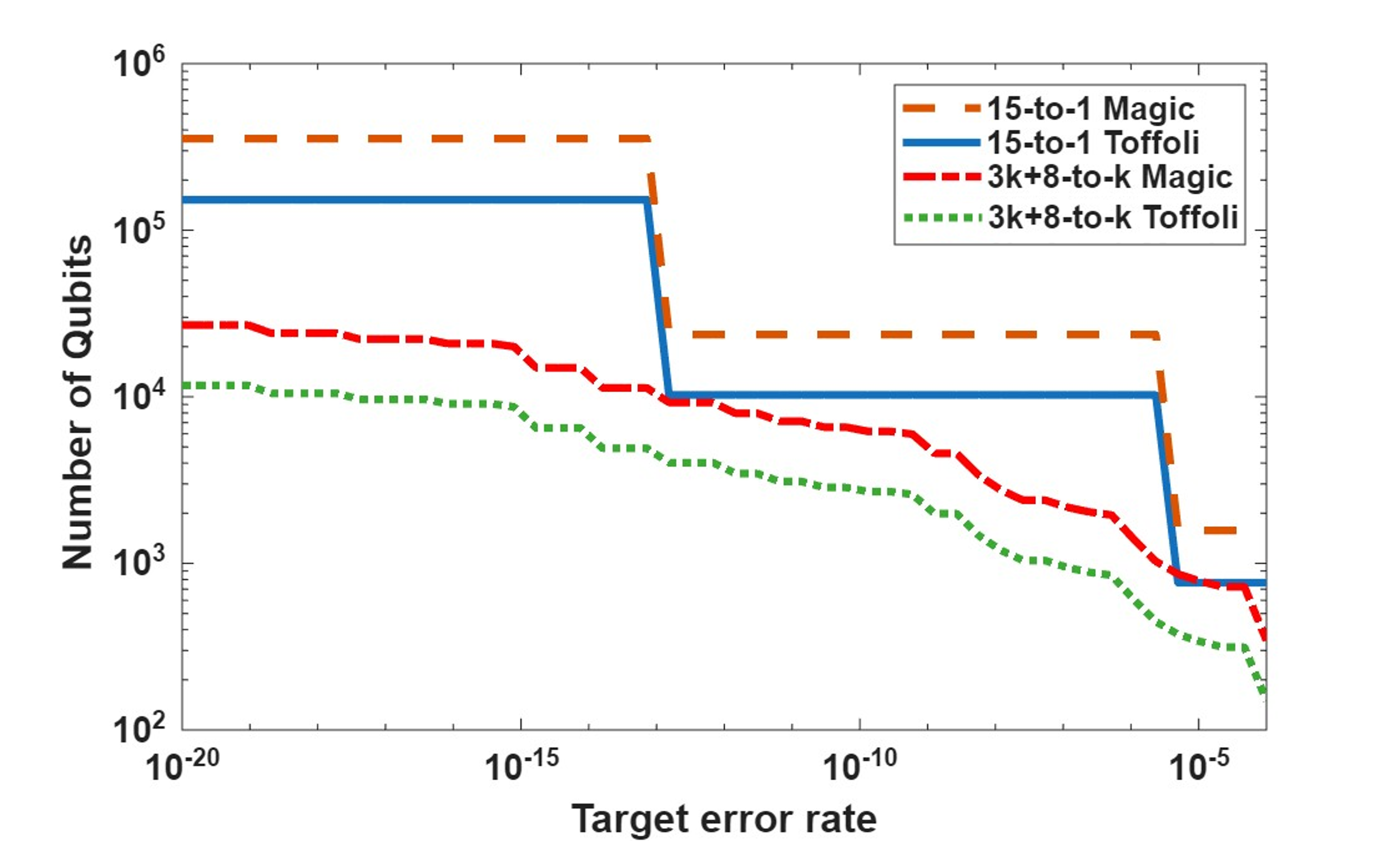}
    \caption{Comparison of qubit resource costs for different Toffoli state distillation protocols. The 15-to-1 and $3k+8$-to-$k$ Toffoli curves corresponding to our direct Toffoli state distillation protocols using the $[[15,1,3]]$ and $[[3k+8,k,2]]$ codes, respectively. The corresponding 'Magic' curves represent the protocols that first distill seven magic states using the same codes to construct a Toffoli state. With the initial physical error rate selected as $10^{-2}$, our direct Toffoli-state distillation approach achieves the same target error rate with approximately $50\%$ fewer qubits. 
    }
    \label{fig.ToffoliCompare}
\end{figure}

We can also estimate the qubit resource cost required to achieve a given level of distilled error rate. In Fig.~\ref{fig.ToffoliCompare}, we compare the qubit overhead required to achieve a target-fidelity Toffoli gate between our direct Toffoli-state distillation protocol and conventional approaches that use magic states to construct Toffoli states. The physical Toffoli gate or magic state error rate is selected as $10^{-2}$, and for $[[3k+8,k,2]]$ code we choose the optimized $k$ for each target error rate.
In our protocol, we directly distill a three-qubit Toffoli state, whereas other approaches use either seven magic states with a 15-to-1 protocol or the more general $3k+8$-to-$k$ distillation protocol.
Because we directly distill the Toffoli state, our method requires only about half as many qubits. Specifically, the average qubit cost for achieving an $m-$level distilled Toffoli state via the 15-to-1 protocol is $$\frac{3\times 15^m+45(m-1)}{p_{suc}},$$ where $p_{suc}$ is the overall successful probability.  In contrast, using magic state distillation requires $7\times15^m/p_{suc}$, which the cost is around $\times 2$ larger than our protocol. 
We can see the same reduce overhead in $[[3k+8,k,2]]$ distillation protocol. Compare with $[[15,1,3]]$ code, $[[3k+8,k,2]]$ with optimized $k$ uses less qubit to achieve target error rate. 

The optimized value of $k$ in the $[[3k+8,k,2]]$ code can be carefully selected to achieve the best performance for a given target error rate. We numerical calculate the qubit number cost for $k=1$ to $50$ with initial physical error rate selected as $10^{-2}$, and the results for $k=2$ to $k=13$ are shown in Fig.~\ref{fig.3k+8}. Other values of $k$ lead to substantially higher overhead within the considered target error-rate range, and therefore omitted. The logical qubit number $k$ at each distillation layer can thus be optimized to minimize the overall qubit overhead.

In addition to reducing the overall resource overhead, our protocol avoids the repeated encoding steps required in other methods.
Since conventional approaches rely on  $\mathbf{T}$-gate-based magic states to construct the Toffoli state, the distilled magic states must be encoded into the computational system for seven times. These encoding procedure may increase the total cost and the logical error rate.
In contrast, our protocol requires only a single encoding step at the end of distillation procedure, making it potentially more robust against CNOT gate errors. A more detailed investigation of this aspect is left for future work.

\begin{figure}
    \centering
    \includegraphics[width=0.95\linewidth]{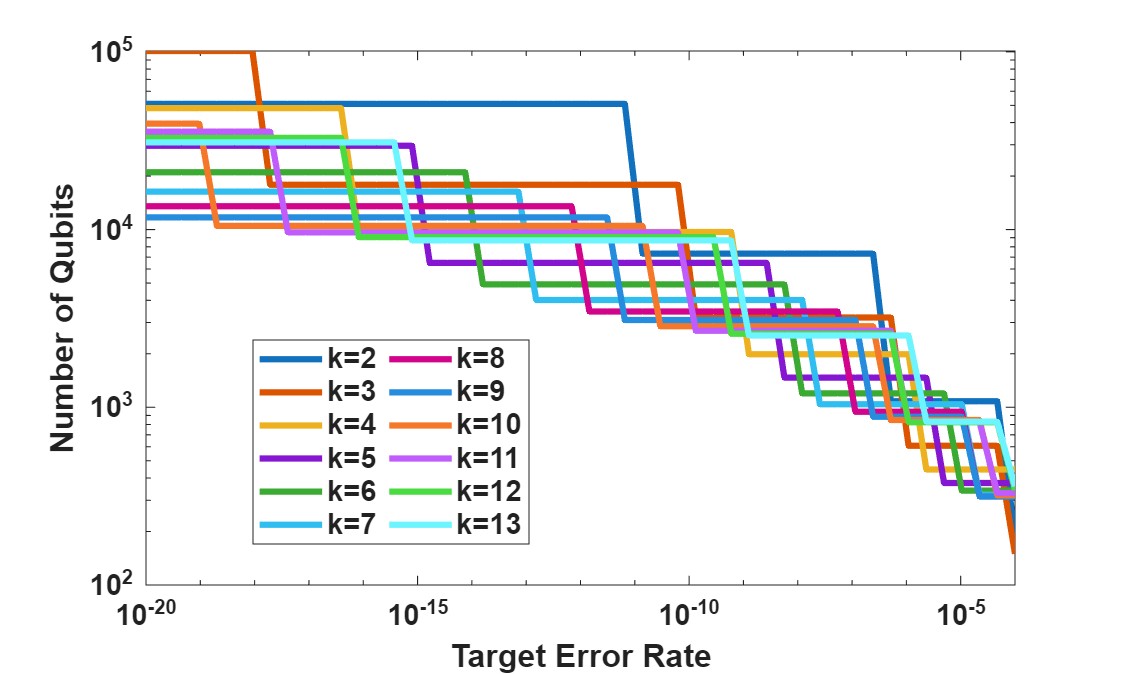}
    \caption{Different values of $k$ in the $[[3k+8,k,2]]$ code. The optimal $k$ is chosen as the value that minimizes the number of qubits required to achieve the target error rate.}
    \label{fig.3k+8}
\end{figure}

\section{Conclusion}
In this paper, we investigated the transversality of the Toffoli gate within a hybrid-code framework. Through circuit-level analysis, we demonstrated that by employing a triorthogonal code and its paired mirrored code in a specific configuration, the Toffoli gate can be implemented transversally. Furthermore, based on this hybrid-code architecture, we proposed a Toffoli gate distillation protocol that does not require pre-distilled $\mathbf{T}$ gates. The numerical results indicate that the proposed method achieves a $50\%$ reduction in the required qubit resources compared to existing approaches in the literature. This design serves as a concrete example of achieving a transversal Toffoli gate. In future work, we aim to establish necessary and/or sufficient conditions for transversality when employing three distinct quantum error correction codes.

\balance
\bibliographystyle{IEEEtran}
\bibliography{ASSref}

@article{haah2018Toffoli,
  title={Codes and protocols for distilling {$T$}, controlled-{$S$}, and {Toffoli} gates},
  author={Haah, Jeongwan and Hastings, Matthew B},
  journal={Quantum},
  volume={2},
  pages={71},
  year={2018},
  publisher={Verein zur F{\"o}rderung des Open Access Publizierens in den Quantenwissenschaften}
}

@article{jones2013Toffoli,
  title={Low-overhead constructions for the fault-tolerant {Toffoli} gate},
  author={Jones, Cody},
  journal={Physical Review A—Atomic, Molecular, and Optical Physics},
  volume={87},
  number={2},
  pages={022328},
  year={2013},
  publisher={APS}
}

@article{eastin2013distilling,
  title={Distilling one-qubit magic states into {Toffoli} states},
  author={Eastin, Bryan},
  journal={Physical Review A—Atomic, Molecular, and Optical Physics},
  volume={87},
  number={3},
  pages={032321},
  year={2013},
  publisher={APS}
}

@article{yoder2017universal,
  title={Universal fault-tolerant quantum computation with {Bacon-Shor} codes},
  author={Yoder, Theodore J},
  journal={arXiv preprint arXiv:1705.01686},
  year={2017}
}

@article{paetznick2013universal,
  title={Universal Fault-Tolerant Quantum Computation with Only Transversal Gates and Error Correction},
  author={Paetznick, Adam and Reichardt, Ben W},
  journal={Physical review letters},
  volume={111},
  number={9},
  pages={090505},
  year={2013},
  publisher={APS}
}

@article{bravyi2005magicstate,
  title={Universal quantum computation with ideal {Clifford} gates and noisy ancillas},
  author={Bravyi, Sergey and Kitaev, Alexei},
  journal={Physical Review A—Atomic, Molecular, and Optical Physics},
  volume={71},
  number={2},
  pages={022316},
  year={2005},
  publisher={APS}
}

@article{eastin2009restrictions,
  title={Restrictions on transversal encoded quantum gate sets},
  author={Eastin, Bryan and Knill, Emanuel},
  journal={Physical review letters},
  volume={102},
  number={11},
  pages={110502},
  year={2009},
  publisher={APS}
}

@article{anderson2014codeswitch,
  title={Fault-tolerant conversion between the steane and {Reed-Muller} quantum codes},
  author={Anderson, Jonas T and Duclos-Cianci, Guillaume and Poulin, David},
  journal={Physical review letters},
  volume={113},
  number={8},
  pages={080501},
  year={2014},
  publisher={APS}
}

@article{horsman2012lattice,
  title={Surface code quantum computing by lattice surgery},
  author={Horsman, Dominic and Fowler, Austin G and Devitt, Simon and Van Meter, Rodney},
  journal={New Journal of Physics},
  volume={14},
  number={12},
  pages={123011},
  year={2012},
  publisher={IOP Publishing}
}

@inproceedings{bayanifar2025transversality,
  title={Transversality Across Two Distinct Quantum Codes and Its Application to Quantum Repeaters},
  author={Bayanifar, Mahdi and Ashikhmin, Alexei and Jiao, Dawei and Tirkkonen, Olav},
  booktitle={IEEE Internat. Conf. on Quantum Commun., Networking, and Computing (QCNC)},
  pages={17--23},
  year={2025},
}

@article{shi2024triorthogonal,
  title={Triorthogonal codes and self-dual codes},
  author={Shi, Minjia and Lu, Haodong and Kim, Jon-Lark and Sol{\'e}, Patrick},
  journal={Quantum Information Processing},
  volume={23},
  number={7},
  pages={280},
  year={2024},
  publisher={Springer}
}

@article{narayana2020optimality,
  title={On optimality of {CSS} codes for transversal {T}},
  author={Rengaswamy, Narayanan and Calderbank, Robert and Newman, Michael and Pfister, Henry D},
  journal={IEEE Journal on Selected Areas in Information Theory},
  volume={1},
  number={2},
  pages={499--514},
  year={2020},
  publisher={IEEE}
}

@book{nielsen2002Book,
  title={Quantum computation and quantum information},
  author={Nielsen, Michael A and Chuang, Isaac L},
  year={2010},
  publisher={Cambridge university press}
}

@article{newman2017limitations,
  title={Limitations on transversal computation through quantum homomorphic encryption},
  author={Newman, Michael and Shi, Yaoyun},
  journal={arXiv preprint arXiv:1704.07798},
  year={2017}
}

@article{Bravyi_2012,
   title={Magic-state distillation with low overhead},
   volume={86},
   ISSN={1094-1622},
   url={http://dx.doi.org/10.1103/PhysRevA.86.052329},
   DOI={10.1103/physreva.86.052329},
   number={5},
   journal={Physical Review A},
   publisher={American Physical Society (APS)},
   author={Bravyi, Sergey and Haah, Jeongwan},
   year={2012},
   month=nov }

@inproceedings{toffoli1980reversible,
  title={Reversible computing},
  author={Toffoli, Tommaso},
  booktitle={International colloquium on automata, languages, and programming},
  pages={632--644},
  year={1980},
  organization={Springer}
}

@article{Gottesman_1998,
   title={Theory of fault-tolerant quantum computation},
   volume={57},
   ISSN={1094-1622},
   url={http://dx.doi.org/10.1103/PhysRevA.57.127},
   DOI={10.1103/physreva.57.127},
   number={1},
   journal={Physical Review A},
   publisher={American Physical Society (APS)},
   author={Gottesman, Daniel},
   year={1998},
   month=jan, pages={127–137} }

@article{vandersypen2001experimental,
  title={Experimental realization of Shor's quantum factoring algorithm using nuclear magnetic resonance},
  author={Vandersypen, Lieven MK and Steffen, Matthias and Breyta, Gregory and Yannoni, Costantino S and Sherwood, Mark H and Chuang, Isaac L},
  journal={Nature},
  volume={414},
  number={6866},
  pages={883--887},
  year={2001},
  publisher={Nature Publishing Group UK London}
}

\end{document}